\documentclass[letterpaper, 10 pt, conference]{ieeeconf}  

\IEEEoverridecommandlockouts                              
\usepackage{fainekos-macros}

\usepackage{tikz}
\graphicspath{ {figures/} }
\DeclareGraphicsExtensions{.pdf}
\usepackage{multirow}
\usepackage{todonotes}
\usepackage{url}
\usepackage{enumerate}
\usepackage{ifthen}
\usepackage{amsmath}
\usepackage{mathtools}

\usepackage{color,soul}
\newboolean{HIGHLIGHTCHANGES}
\setboolean{HIGHLIGHTCHANGES}{false}
\newcommand \yhl[1]{\ifthenelse{\boolean{HIGHLIGHTCHANGES}}{\textcolor{blue}{#1}}{#1}}

\newcommand{\slabel}{\sigma}
\newcommand{\tlabel}{\slabel_\chi}
\newcommand{\ulabel}{\slabel_u}
\newcommand{\labelSet}{\Sigma}
\newcommand{\ulabelSet}{\labelSet_u}
\newcommand{\bslabel}{\bar{\slabel}}
\newcommand{\tlabelSet}{\labelSet_\chi}
\newcommand{\btlabel}{\bar{\tlabel}}
\newcommand{\trans}[1]{\xrightarrow{#1}}
\newcommand{\out}[1]{\left\langle #1\right\rangle}

\newboolean{REPORT}
\setboolean{REPORT}{TRUE}

\title{\LARGE \bf
Towards composition of conformant systems
}

\author{Houssam Abbas and Georgios Fainekos
\thanks{H. Abbas is with the Department of Electrical, Computer and Energy Engineering, Arizona State University,
		Tempe, U.S.A.
        {\tt\small hyabbas@asu.edu}}%
\thanks{G. Fainekos is with the School of Informatics, Decisions and Systems Engineering, Arizona State University,
		Tempe, U.S.A.        
        {\tt\small gfaineko@asu.edu}}%
\thanks{This work was partially supported by NSF awards CNS 1350420 and CPS 1446730.}
}

\begin{document}
	
\maketitle
\thispagestyle{empty}
\pagestyle{empty}
\begin{abstract}

Motivated by the Model-Based Design process for Cyber-Physical Systems, we consider issues in conformance testing of systems.
Conformance is a quantitative notion of similarity between the output trajectories of systems, which considers both temporal and spatial aspects of the outputs.
Previous work developed algorithms for computing the conformance degree between two systems, and demonstrated how formal verification results for one system can be re-used for a system that is conformant to it.
In this paper, we study the relation between conformance and a generalized approximate simulation relation for the class of Open Metric Transition Systems (OMTS).
This allows us to prove a small-gain theorem for OMTS,  which gives sufficient conditions under which the feedback interconnection of systems respects the conformance relation, thus allowing the building of more complex systems from conformant components.
\end{abstract}

\section{Introduction}
In Model-Based Design (MBD) of systems, an executable model of the system is developed early in the design process.
This allows the verification engineers to conduct early testing~\cite{Butts_ToyotaMBD06}.
The model is then refined iteratively and more details are added, e.g., initially ignored physical phenomena, time delays, etc.
This eventually leads to the final model that gets implemented on some computational platform, for example via automatic code generation.
See Fig.~\ref{fig:MBD_v_process}.

Each of the above transformations and calibrations introduces discrepancies between the output behavior of the original system (\textbf{the nominal system}) and the output behavior of the derived system (\textbf{the derived system}).
These discrepancies are spatial (e.g., slightly different signal values in response to same stimulus, dropped samples, etc) and temporal (e.g., different timing characteristics of the outputs, out-of-order samples, delayed responses, etc) and their magnitude can vary as time progresses.

Ideally, the initial (simpler) model should be amenable to formal synthesis and verification methods (cycle 1 in Fig.~\ref{fig:MBD_v_process}) through tools like \cite{FrehseCAV11,PlatzerQ08ijcar}.
To understand how the formal verification results on the simpler nominal model can be applied to the derived more complex system, it is necessary to quantify the \emph{conformance degree} between them.
The conformance degree, introduced in \cite{AbbasHFDKU_ConformanceArxiv14,AbbasMF_MemocodeConformance14}, is a measure of both spatial and temporal differences between the output behaviors of two systems.
It relaxes traditional notions of distance, like sup norm and approximate simulation, to encompass a larger class of systems, and to allow re-ordering of output signal values.
In \cite{AbbasMF_MemocodeConformance14}, it was shown how the formal properties satisfied by the derived system can be automatically obtained from knowledge of the properties satisfied by the nominal system, and knowledge of the conformance degree between them.
\begin{figure}[t]
	\centering
	\includegraphics[width=\columnwidth]{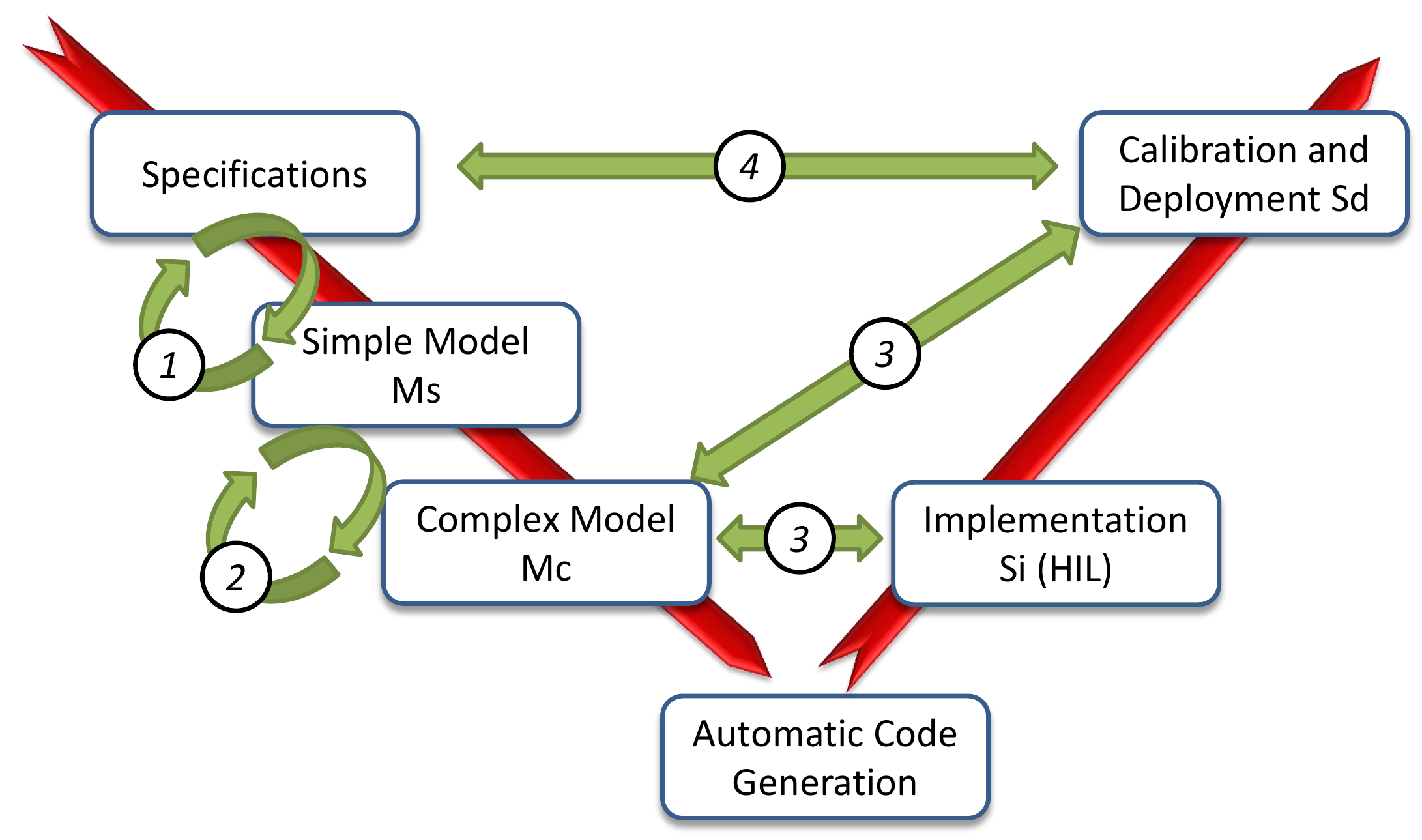}
	\caption{Model-Based Development V-process.}
	\label{fig:MBD_v_process}
\end{figure}
In this paper, we extend that work by studying feedback interconnections of systems.
Specifically, we are concerned with the following question: suppose we have a feedback interconnection of a plant and controller, and the closed-loop system has been formally verified to satisfy some properties.
If the controller (or the plant) is replaced by another controller which is conformant to it, is the new closed-loop system conformant to the original closed-loop system? 
If yes, can we estimate its conformance degree without explicitly re-computing it?
A positive answer to both questions would allow us to leverage the results in \cite{AbbasMF_MemocodeConformance14} and automatically deduce the properties satisfied by the new interconnection. 

In this paper, we give a positive answer to both questions for a general class of dynamical systems modeled as \emph{Open Metric Transition Systems} (OMTS). 
These are defined in Section \ref{sec:OMTS}.
The tool we use is a generalized notion of \emph{Space-Time Approximate Simulation} (STAS) relation, which is defined in Section \ref{STAS}.
We show in Section \ref{sim2closenessMTS} that the existence of such a relation between two OMTS implies that they are also conformant, and yields the conformance degree between them.
In Section \ref{sec:interconnectingConfSys} we provide a small-gain theorem for OMTS, which gives sufficient conditions under which feedback interconnections of OMTS respect approximate simulation, and therefore conformance.
This is done via \emph{STAS functions}, which are Lyapunov-like functions that certify the existence of a STAS relation between two systems.

\textbf{Notation.} For a positive integer $n$, $[n] = \{1,\ldots,n\}$.
Given a set $\labelSet$, $\labelSet^*$ is the set of finite strings on $\labelSet$, i.e. $\labelSet^* = \{s_0s_1\ldots s_n \sut s_i \in \labelSet, n \in \Ne\}$.
Given two sets $A,B$ and $(a,b) \in A\times B$, $\proj{(a,b)}{A}=a$.

\section{Conformance of Open Metric Transition Systems}
\label{sec:closeness2simulation}
In this section, we define a general system model, namely, Open Metric Transition Systems (OMTS).
These extend Metric Transition Systems \cite{GirardP07tac} in that they allow interconnection of systems, and will be our formalism of choice in this paper.
We then define the conformance relations for OMTS and feedback interconnections for OMTS, which allows us to speak of controlled OMTS and compositionality in Section \ref{sec:interconnectingConfSys}.
As an illustration, we show how hybrid systems can be modeled as OMTS.

\subsection{Open metric transition systems and conformance}
\label{sec:OMTS}
A Metric Transition System (MTS) serves to model, at an abstract level, a fairly large class of systems.
An MTS is a tuple $T = (Q,\tlabelSet,\rightarrow,Q^0,\Pi,\out{\;})$ 
where
$Q$ is a set known as the state space, 
$Q^0 \subset Q$ is the set of initial states, 
$\tlabelSet$ is the set of labels on which transitions take place, 
$\trans{\;} \subset Q\times \tlabelSet \times Q$ is the transition relation, 
$\Pi$ is the output set, 
and $\out{}: Q \rightarrow \Pi$ is the output map. 
We write $q \trans{\tlabel} q'$ to denote an element $(q,\tlabel,q')\in \rightarrow$.
Both $Q$ and $\Pi$ are metric spaces, that is, they are equipped with metrics $d_Q$ and $d_\Pi$. 
Moreover, for any $q \in Q$ and any label subset $S \subset \tlabelSet$, the set 
\begin{equation}
\label{eq:closedPostset}
\text{Post}(q,S) = \cup_{\tlabel \in S} \text{Post}(q,\tlabel)
\end{equation} 
is compact in the metric-induced topology.

Given a string of labels $\btlabel = {\tlabel}_{,1}{\tlabel}_{,2}\ldots{\tlabel}_{,m}$, we write ${\btlabel}_{[i]} \in \tlabelSet^*$ for the prefix string ${\tlabel}_{,1}{\tlabel}_{,2}\ldots{\tlabel}_{,i}$, $i \leq m$.
The sets $\tlabelSet$ and $\tlabelSet^*$ are equipped with pseudo-metrics\footnote{A \emph{pseudo-metric} does not separate points.} $d_{\tlabelSet}$ and $d_{\tlabelSet^*}$, respectively, and $\Pi$ is equipped with a metric $d_\Pi$.
When two MTS share the same $\tlabelSet$ ($\tlabelSet^*$, $\Pi$), they also share the same associated (pseudo-)metrics.

An \emph{Open} Metric Transition System (OMTS) is a tuple $T = (Q,\Sigma,\trans{}_T,Q^0,\Pi,\out{\;},p)$ 
where $(Q,\Sigma,\rightarrow_T,Q^0,\Pi,\out{\;})$ is an MTS as above.
The label set $\Sigma$ of an OMTS has a special structure: $\labelSet \subset \labelSet_u \times \labelSet_\chi$ for sets $\ulabelSet, \tlabelSet$.
The intuition behind this division is that $\ulabelSet$ will be used to model input signals to the system embedded as an OMTS, and $\tlabelSet$ will be used to model the domain of that input signal.
This departs from earlier approaches to embedding forced dynamical systems as MTS~\cite{Girard2005b}, because we need a way to describe interconnections of MTS, while preserving timing information in the interconnection.
A generic label $\slabel$ thus has two components: $\slabel = (\ulabel, \tlabel)$.
The string prefix $\bslabel_{[i]}$ is defined similarly to the case of MTS.
The \emph{port map} $p: (\trans{}_T) \rightarrow \labelSet \cup \{\nu\}$ associates a label to each transition in $\trans{}_T$, or a special \emph{empty label $\nu$}.
The empty label, as we will see, is used to allow a system to make empty transitions which don't change its state and don't advance time.
The output of the port map will be used to compose OMTS.
This makes them similar to hybrid I/O automata \cite{LynchSV03_HybIOAutomata} but enriched with a metric structure, and with `discrete actions' and `trajectories of input variables' lumped into one label set, which fits well our usage of hybrid time.

We now define conformance between two OMTS $T_1$ and $T_2$.
Conformance quantifies the similarity between systems, and accounts for the fact that in a typical MBD process (Fig.~\ref{fig:MBD_v_process}), the output signals of the derived model will have temporal and spatial differences with the outputs of the nominal model.
From the knowledge of the conformance degree between two systems, we can conclude what formal specifications are satisfied by one, given the specifications satisfied by the other \cite{AbbasMF_MemocodeConformance14}.
\begin{definition}[Conformance]
	\label{def:conformanceMTS}
	Let $T_1$ and $T_2$ be two OMTS with a common output space $\Pi$ and common label set $\labelSet$.
	Let $\tau, \varepsilon$ be two non-negative reals.
	Let $D \subset Q_1^0 \times Q_2^0$ be a relation defined on their initial sets. 
	We refer to $D$ as the \emph{derivation relation}.	
	We say $T_2$ \textbf{conforms to $T_1$ with precision $\mathbf{\teps}$ and derivation relation $D$}, 
	which we write $T_1 \teconforms T_2$, if for all $(q_1^0,q_2^0)\in D$, 
	and any sequence of $T_1$ transitions
	\[q_1^0 \trans{\slabel_1}_1 q_1^1 \trans{\slabel_2}_1 q_1^2 \trans{...}_1 \ldots \trans{\slabel_n}_1 q_1^m\]
	there exists a sequence of $T_2$ transitions
	\[q_2^0 \trans{\alpha_1}_2 q_2^1 \trans{\alpha_2}_2 q_2^2 \trans{...}_2 \ldots \trans{\alpha_n}_2 q_2^{m'}\]	
	 such that 
	 \begin{enumerate}[(a)]
	 	\item for all $q_1^i$, $i\in [m]$, there exists $q_2^k$ s.t. $d_\Pi(\out{q_1^i},\out{q_2^k}) \leq \varepsilon$ and $d_{\labelSet^*}(\bslabel_{[i]}, \bar{\alpha}_{[k]}) \leq \tau$ 
	 	\label{item:confa}
	 	\item for all $q_2^i$, $i \in [m']$, there exists $q_1^k$ s.t. $d_\Pi(\out{q_1^k},\out{q_2^i}) \leq \varepsilon$ and $d_{\labelSet^*}(\bslabel_{[k]}, \bar{\alpha}_{[i]}) \leq \tau$ 
	 	\label{item:confb}
	 \end{enumerate}
	 
	\end{definition}
Intuitively, the definition requires $T_2$ to be able to match any execution of $T_1$, with some allowed deviation between the states that each execution visits, and some allowed deviation between the labels on which transitions take place. 
The matching is required not only for the final reached states $q_1^m$ and $q_2^{m'}$, but for all intermediary states.
The relation $D$ is meant to capture the mapping between the initial states of one model ($T_1$) and the initial states of its implementation ($T_2$).
For example, if $T_2$ is obtained by model order reduction from $T_1$, $D$ captures the reduction mapping as applied to the initial states.
Because some of the labels in either transition sequence may be the empty label $\nu$, more than one state in one sequence may match with the same state in the other sequence.

\subsection{Feedback interconnection of OMTS}
\label{sec:feedbackOMTS}

Given two OMTS $T_1$ and $T_2$, we define their feedback interconnection as follows.
\begin{definition}[Feedback in OMTS]
\label{def:feedbackOMTS}
Let $T_i$ be an OMTS 
$(Q_i,\labelSet,\trans{\;}_i, Q_i^0, \Pi_i, \out{\;}_i, p_i)$, $i=1,2$, 
such that $\labelSet = \ulabelSet \times \tlabelSet$.
Assume that $\labelSet \supset p_2(\trans{\;}_2)$ and
$\labelSet \supset p_1(\trans{\;}_1)$.
Their \emph{feedback interconnection} is a (closed) MTS 
$(Q,{\tlabelSet}_{12},\trans{\;}, Q^0, \Pi, \out{\;})$, 
denoted $T_1 \circ T_2$, where
\begin{itemize}
	\item $Q = Q_1 \times Q_2$
	\item ${\tlabelSet}_{12} \subset \tlabelSet$
	\item $Q^0 = Q_1^0 \times Q_2^0$
	\item $\Pi = \Pi_1 \times \Pi_2$
	\item $\out{(q_1,q_2)} = (\out{q_1}_1, \out{q_2}_2)$
	\item $\trans{}$: $(q_1,q_2) \trans{\tlabel} (q_1',q_2')$ iff $\exists \slabel_1 = (\slabel_{1,u},\tlabel) \in \labelSet$
	and $\slabel_2 = (\slabel_{2,u},\tlabel) \in \labelSet$ 
	s.t. $q_1 \trans{\slabel_1}_1 q_1'$, $q_2 \trans{\slabel_2}_2 q_2'$, 
	and $\slabel_{1} = p_2(q_2 \trans{\slabel_2}_2 q_2')$, $\slabel_{2} = p_1(q_1 \trans{\slabel_1}_1 q_1')$.
\end{itemize}

The output set distance is given by 
\[d_\Pi((q_1,q_2),(q_1',q_2')) = \tilde{h}(d_{\Pi_1}(q_1,q_1'), d_{\Pi_2}(q_2,q_2'))\]
 for some positive non-decreasing function $\tilde{h}$.
\end{definition}
This is meant to model the situation when two hybrid systems are feedback interconnected, such that $T_1$'s outputs constitute the inputs to $T_2$, and vice versa.
Note that the definitions of output set, output map and associated distance function are somewhat arbitrary and ultimately depend on the application domain.

To simplify the statement of the main theorem and its proof, we introduce the following `lifting' of ${\tlabelSet}_{12}$ to $\labelSet \times \labelSet$.
The set $\labelSet_{12}$ defined below contains all label pairs  $(\slabel_1,\slabel_2) \in \labelSet_1 \times \labelSet_2$ allowed by the interconnection $T_1 \circ T_2$.
Formally:
\begin{equation}
\label{eq:lifting}
\begin{split}
\labelSet_{12} \defeq \{&(\slabel_1,\slabel_2) \in \labelSet \times \labelSet \sut \slabel_1 = p_2(q_2 \trans{\slabel_2}_2 q_2'),\\
&\slabel_2 = p_1(q_1 \trans{\slabel_1}_1 q_1'), {\tlabel}_1 = {\tlabel}_2 \in {\tlabelSet}_{12},\\
&\text{ for some transitions } q_2 \trans{\slabel_2}_2 q_2' \text{ and } q_1 \trans{\slabel_1}_1 q_1'\}
\end{split}
\end{equation}
We note two properties of $\labelSet_{12}$:
\begin{enumerate}
	\item $\labelSet_{12} \subset \labelSet \times \labelSet$
	\label{item:label1}
	\item minimizing a function over the transitions enabled by labels in ${\tlabelSet}_{12}$ 
	yields the same result as minimizing it over the transitions enabled by labels in the lifting $\labelSet_{12}$.
	\label{item:label2}
\end{enumerate}

\subsection{Problem formulation}
\label{sec:problemformulation}
The formal statement of this paper's problem follows:

Given two OMTS $T_1$ and $T_2$ connected in a feedback loop, and OMTS $T_3$ that conforms to $T_1$ with precision $\mathbf{\teps}$ and derivation relation $D$, is $T_3 \circ T_2$ conformant to the $T_1 \circ T2$?
If yes, what is the conformance degree between the two loops?

\subsection{Embedding a hybrid system as an OMTS}
\label{sec:embeddingHSasOMTS}
Hybrid systems can be represented using, or embedded as, OMTS. 
This enables us to apply the compositionality result to them.
We briefly define hybrid systems to show the embedding.
Let $C$ and $D$ be subsets of $\Re^{n+m}$, $U \subset \Re^m$ be a set of input values, 
\yhl{$F: \Re^{n+m} \rightrightarrows \Re^n$ and $G:\Re^{n+m}\rightrightarrows \Re^n$} be set-valued maps with $C \subset \dom F$ and $D \subset \dom G$.
Let $z: \Re^n \rightarrow \Re^{n_z}$ be a function.
The \emph{hybrid dynamical system} $\Sys$ with data $(C,F,D,G,z)$, internal state $\hsPt \in \Re^n$ and output $\hoPt \in \Re^{n_z}$ is governed by \cite{GoebelT06_SolnsHybInclusions}
\begin{equation}
\label{eq:HA}
\Sys \left\{ \begin{array}{lll}
\dot{\hsPt}  & \in F(\hsPt, \inpPt)  &\quad (\hsPt, \inpPt) \in C \\
\hsPt^+      &\in  G(\hsPt, \inpPt)  		&\quad (\hsPt, \inpPt) \in D   \\
\hoPt 	     &= z(\hsPt)          		
\end{array}
\right.
\end{equation}
The `jump' map $G$ models the change in system state at a mode change, or `jump', and the jump set $D$ captures the conditions causing a jump.
The `flow' map $F$ models state evolution away from jumps, while $(\hsPt,\inpPt)$ is in the flow set $C$. 
System trajectories start from a specified set of initial conditions $\hsSet_0  \yhl{\subset \proj{\overline{C}\cup D}{\Re^n}}$.
Finally, the output of the system $\hoPt$ is given as a function $z$ of its internal state,
and its input is given by $\inpPt$ which takes values in a set $\inpSet$.

Solutions $(\hstraj,\inpSig$) to \eqref{eq:HA} are given by a hybrid arc $\hstraj$ and an \emph{input arc} $\inpSig$ sharing the same hybrid time domain $\dom \hstraj = \dom \inpSig$, 
and with standard properties that can be reviewed in~\cite[Ch. 2]{GoebelST_HybridSysBook} . 
\begin{definition}[Hybrid time domains and arcs~\cite{GoebelT06_SolnsHybInclusions}]
	A subset $E \subset \Re_+\times \Ne$ is a \emph{compact hybrid time domain} if 
	\[E = \bigcup_{j=0}^{J-1}[t_j,t_{j+1}]\times \{j\}\]
	for some finite increasing sequence of times $0=t_0 \leq t_1 \leq t_2 \leq \ldots \leq t_J$.
	A \emph{hybrid arc} $\hstraj$ is a function supported over a hybrid time domain $\hstraj:E \rightarrow \Re^n$, 
	such that for every $j$, $\hstraj(\cdot,j)$ is locally absolutely continuous in $t$ over $I_j = \{t: (t,j) \in E\}$;
	we call $E$ the domain of $\hstraj$ and write it $\dom \hstraj$.
\end{definition}	
	
A hybrid system $\Sys = (C,F,D,G,z)$ can be embedded as an OMTS 
$T = (Q,\Sigma,\rightarrow,Q^0,\Pi,\out{\;},p)$ as follows:
$Q = \{x \in \Re^n \sut \exists u: (x,u) \in \overline{C}\cup D \}$,
$Q^0 \subset Q$,
$\out{\;} = z$,
and $\Pi = \Re^{n_z}$.
The label set is made of input arcs and their domains, and the empty label $\nu$:
\begin{equation}
\label{eq:SigmaEmbedding}
\Sigma = \{(\inpSig,\dom \inpSig) \sut \inpSig \text{ is an input arc}\} \cup \{\nu\}
\end{equation} 

The transition relation is defined as $q \trans{\slabel}q'$ iff either $\slabel = \nu$ is the empty label and $q=q'$,
or $\slabel = (\inpSig,\dom \inpSig)$ and there exists a solution pair $(\hstraj,\inpSig)$ s.t. 
$\hstraj(0,0)=q, \hstraj(t,j) = q'$ for some $(t,j)$ in $\dom \inpSig$.
The port map $p$ is defined as 
\begin{eqnarray*}
p(q \trans{\sigma} q') = \left\{\begin{matrix}
(z(q),(0,0)) & if \sigma = \nu\\
(z\circ \hstraj,\dom \hstraj) & \text{ otherwise}
\end{matrix}\right.
\end{eqnarray*}
 where $(\hstraj,\inpSig)$ is the solution pair of $\Sys$ corresponding to $\sigma$ as defined above in \eqref{eq:SigmaEmbedding}.

Later in the paper, we will need to impose a requirement on $d_\labelSet$, namely, equation  \eqref{eq:labelPseudo} from Section \ref{sim2closenessMTS}. 
The rest of this section shows how $d_\labelSet$ can be defined so this requirement is met.
First, given an input arc $\inpSig$ with domain $E$ and two subsets $E' \subset E$, $E'' \subset E$ such that $(0,0) \in E'\cap E''$ and $\sup_j E' = \sup_j E''$, the restrictions of $\inpSig$ to $E'$ and $E''$ respectively are said to have a common extension.
(So the restricted arcs start at $(0,0)$ and make the same number of jumps).

Let $\slabel = (\inpSig,E),\slabel' = (\inpSig',E')$ be two labels with $E = \cup_j^{J-1} I_j\times \{j\}$, $E' = \cup_j^{J'-1} I'_j\times \{j\}$ compact hybrid time domains with $J$ and $J'$ jumps, respectively. 
Define 
\begin{eqnarray*}
	d_\labelSet(\slabel,\slabel') \defeq \left\{\begin{matrix}
	\max_jd_H(I_j,I_j') && \text{$\inpSig$ and $\inpSig'$ have a}
	\\ 
	\;&&\text{common extension}
	\\
	\infty && \text{ otherwise }
	\end{matrix}\right.
\end{eqnarray*}
Here, $d_H$ is the symmetric Haussdorff distance between two sets.
A string  $s= \slabel_1 \slabel_2 \ldots \slabel_m$ is then a concatenation of the input arcs and their hybrid time domains\footnote{The concatenation of two compact hybrid time domains $E = \bigcup_{j=0}^{J_1-1}([t_j,t_{j+1}]\times j)$ and 
	$E' = \bigcup_{j=0}^{J_2-1}([t'_j,t'_{j+1}]\times j)$ is the hybrid time domain 
	$E_c = \bigcup_{j=0}^{J_1-1}([t_j,t_{j+1}]\times j) \cup \bigcup_{j=0}^{J_2-1}([t'_j+t_{J_1},t'_{j+1}+t_{J_1}]\times \{j'+N_1\})$} 
, and is itself a valid pair (input arc, hybrid time domain).
That is, in this case, $\labelSet^* \subset \labelSet$.
Therefore given two strings $s$ and $a$, we simply define $d_{\labelSet^*}(s,a) = d_{\labelSet}(s,a)$.
It can be shown that this satisfies \eqref{eq:labelPseudo}.
\section{From simulation relations to conformance relations}
\label{sec:simu2conformance}
\subsection{Space-Time Approximate Simulations}
\label{STAS}	
A \emph{Space-Time Approximate Simulation} (STAS) relation is an approximate simulation relation in the sense of \cite{JuliusP_ApxSynchronizationMTS06}.
We choose to introduce the new terminology in order to avoid potentially awkward (and possibly confusing) references to `simulation relations in the sense of [xyz]'. 
STAS were introduced in~\cite{JuliusP_ApxSynchronizationMTS06} and applied in~\cite{JuliusIBP_ApxSynchronizationMTS09} to the study of networked control systems. 

Our interest in this paper is on conformance as defined earlier, which is a notion defined on entire trajectories.
STAS relations, defined on individual states of systems, is a related notion which has the advantage of having a functional characterization, much like Lyapunov functions characterize stability.
In this section, we define STAS relations and connect them to conformance.
The functional characterization of STAS can then be used to characterize conformance.
\begin{definition}[STAS]
\label{def:simulation}
Given two OMTS $T_i = (Q_i,\labelSet, \trans{}_i,Q_i^0, \Pi, \out{}_i,p_i),i=1,2$, and positive reals $\tau, \varepsilon$,
consider a relation $R \subset Q_1 \times Q_2$, and the following three conditions: 
\begin{enumerate}
	\item $\forall (q_1,q_2) \in R$, $d_\Pi(\out{q_1},\out{q_2}) \leq \varepsilon$
	\label{item:sim1}
	\item $\forall (q_1,q_2) \in R$, $\forall q_1 \trans{\slabel_1 \in \labelSet} q_1'$, $\exists \slabel_2 \in B_\tau(\slabel_1)$ and a transition $q_2 \trans{\slabel_2} q_2'$ s.t. $(q_1',q_2') \in R$
	\label{item:sim2}
	\item $\forall q_1^0 \in Q_1^0, \exists q_2^0 \in Q_2^0$ s.t. $(q_1^0, q_2^0) \in R$
	\label{item:sim3}
\end{enumerate}
where $B_\tau(\slabel) = \{\slabel' \in \labelSet \sut d_\labelSet(\slabel,\slabel')\leq \tau\}$.
If $R$ satisfies the first 2 conditions, then it is a $\mathbf{(\tau,\varepsilon)}$\textbf{-space-time approximate simulation} (STAS) of $T_1$ by $T_2$.
If in addition it satisfies the third, then we say $T_2$ \textbf{simulates} $T_1$ \textbf{with precision} $\teps$.
\end{definition}
STAS relations describe what happens when $T_1$ `plays' label $\sigma_1$, and $T_2$ is allowed to respond by playing a label from $B_\tau(\slabel_1)$.
In particular, it says that $T_2$ can always find a label such that the distance between the reached outputs is less than $\varepsilon$.
In the rest of this paper, we will often simply speak of a simulation to mean a STAS.
%

\subsection{From simulation to conformance}
\label{sim2closenessMTS}
The connection between STAS, which is a relation between states, and conformance, which is a relation between executions, is captured in the following proposition.

\begin{prop}
	\label{prop:sim2closenessMTS}
	Given two OMTS $T_i = (Q_i,\labelSet, \trans{}_i, \Pi, \out{}_i,g_i),i=1,2$,
	let $R$ be a $\teps$-STAS relation between them,
	and let $D \subset Q_1^0 \times Q_2^0$ be a derivation relation between them.
	Assume that the label pseudo-metrics $d_\labelSet$, $d_{\labelSet^*}$ are such that for any two strings $\bar{\slabel} = \slabel_1 \ldots \slabel_i$ and $\bar{\alpha} = \alpha_1 \ldots \alpha_i$, 
	\begin{equation}
	\label{eq:labelPseudo}
		(\forall k \leq i, d_\labelSet(\slabel_k,\alpha_k) \leq \tau) \; \implies d_{\labelSet^*}(\bar{\slabel}_{[i]}, \bar{\alpha}_{[i]}) \leq \tau
	\end{equation}
	If $D \subset R$, then $T_2$ conforms to $T_1$ with precision $\teps$ and with derivation relation $D$.
	\exmend
	\end{prop}
	
\begin{proof}
	Take any pair $(q_1^0,q_2^0) \in D$, and any sequence of $T_1$ transitions 
	\[q_1^0 \trans{\slabel_1}q_1^1 \trans{\slabel_2}q_1^2 \trans{\slabel_3}\ldots \trans{\slabel_n}q_1^n\]
	Because $D \subset R$, there exists a $T_2$ transition $q_2^0 \trans{\alpha_1}q_2^1$ s.t. $\alpha_1 \in B_\tau(\slabel_1)$ and $(q_1^0,q_2^0) \in R$, therefore $d_\Pi(\out{q_1^1},\out{q_2^1}) \leq \varepsilon$.
	Proceeding in this way for every $k \leq n$, we build a sequence $\mathbf{q_2}$ of $T_2$ transitions
	\[\mathbf{q_2} = q_2^0 \trans{\alpha_1}q_2^1 \trans{\alpha_2}q_2^2 \trans{...}\ldots \trans{\alpha_n}q_2^n\]
	such that $d_\labelSet(\slabel_k,\alpha_k) \leq \tau$ and $d_\Pi(\out{q_1^k},\out{q_2^k}) \leq \varepsilon$ for all $k$.
	
	Now we check condition (\ref{item:confa}) of Def.\ref{def:conformanceMTS}. 
	For any $q_1^i,i \leq n$, $d_\Pi(\out{q_1^i},\out{q_2^i}) \leq \varepsilon$ and by property \eqref{eq:labelPseudo} of the label pseudo-metric, $d_{\labelSet^*}(\slabel_{[i]},\alpha_{[i]}) \leq \tau$.
 	Thus condition (\ref{item:confa}) is satisfied. 
 	By construction of the execution $\mathbf{q_2}$ and symmetry of $d_\Pi$ and $d_{\labelSet}$, we also have condition (\ref{item:confb}).
	\end{proof}

\section{Compositionality}
\label{sec:interconnectingConfSys}
In this section, we prove a general small gain condition under which the feedback interconnection of OMTS preserves similarity relations.
By Prop. \ref{prop:sim2closenessMTS}, this implies that conformance is also preserved under these conditions.
We work in the OMTS formalism as it bypasses unnecessary technicalities and allows us to establish the result in greater generality, while maintaining continuity with the work of \cite{JuliusIBP_ApxSynchronizationMTS09}.

\subsection{Compositionality of similar metric transition systems}
\label{sec:interconnectingMTS}

\newcommand \qot {q_{12}}
\newcommand{\qtf}{q_{34}}
\newcommand{\vot}{V_{13}}
\newcommand{\vtf}{V_{24}}
\newcommand{\bVot}{\underbar{V}_{13}}
\newcommand{\bVtf}{\underbar{V}_{24}}
\newcommand{\bV}{\underbar{V}}
\newcommand{\lblSetc}[1]{{\tlabelSet}_{#1}}
\newcommand{\Bttf}{B_\tau^{34}(\slabel_{12})}

Consider OMTS $T_1,T_2,T_3,T_4$ with label sets $\labelSet_1= \labelSet_2$ and $\labelSet_3 = \labelSet_4$.
\begin{figure}
	\centering
		\includegraphics[width=6.5cm]{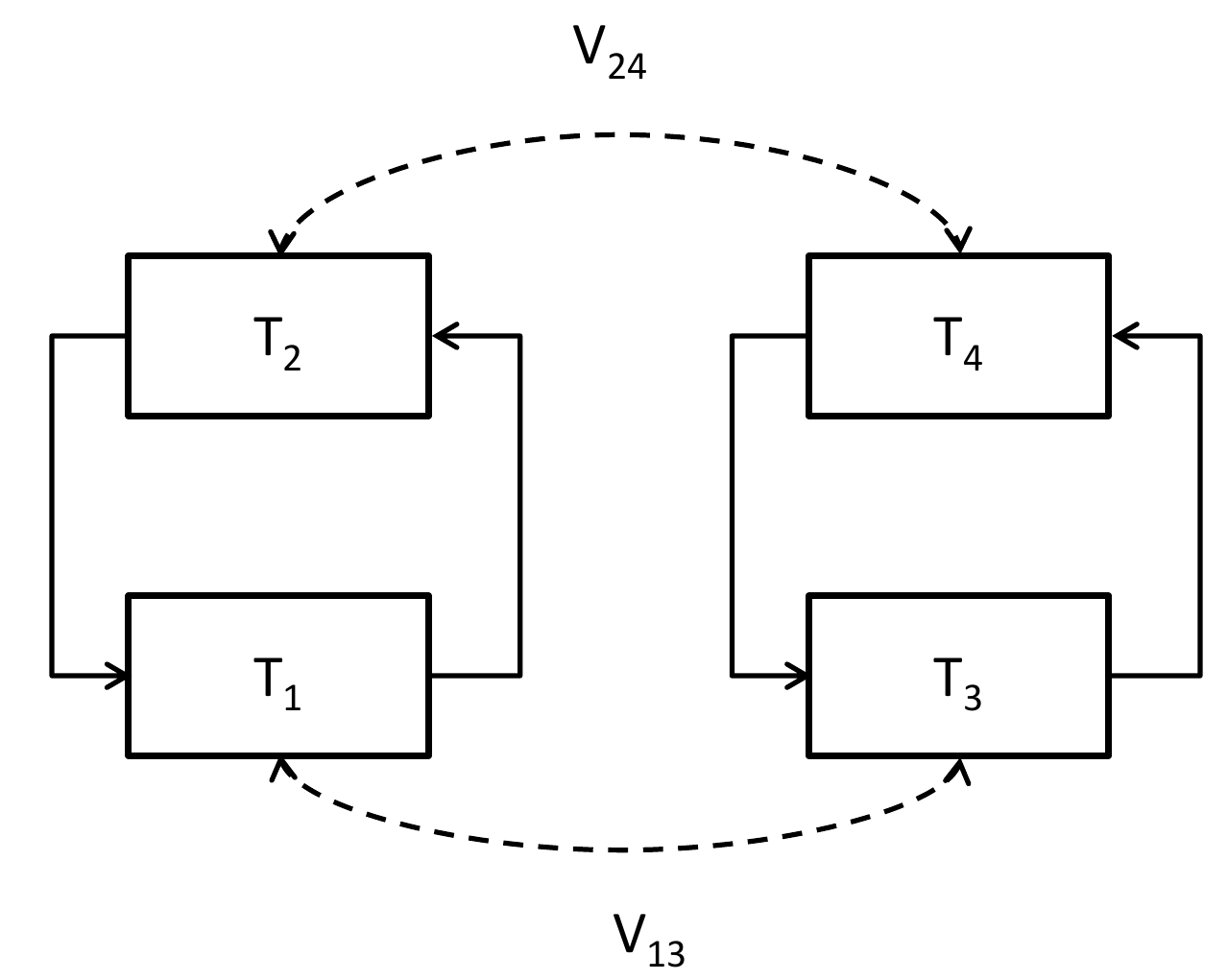}
	\caption{Interconnections of similar MTS}
	\label{fig:interconnected}
\end{figure}
Systems $T_1$ and $T_2$ are feedback interconnected to yield $T_1 \circ T_2$, 
with state space $Q_{12} = Q_1\times Q_2$, and label set $\lblSetc{12}$.
Similarly, systems $T_3$ and $T_4$ are feedback interconnected to yield $T_3 \circ T_4$,
with state space $Q_{34} = Q_3\times Q_4$, and label set $\lblSetc{34}$.
See Fig.~\ref{fig:interconnected}.
We seek conditions under which $T_3 \circ T_4$ simulates $T_1 \circ T_2$;
based on Prop.\ref{prop:sim2closenessMTS}, this would imply that under the same conditions, $T_1 \circ T_2 \teconforms T_3 \circ T_4$ for some $\teps$.
To do so, we use the functional characterization of STAS.
\begin{definition}\cite[Def. 3.2]{JuliusP_ApxSynchronizationMTS06}
\label{def:stas fnt omts}
Given two OMTS $T_1$ and $T_2$ with common output set $\Pi$ and label set $\labelSet$,
and non-negative real $\tau$,
a function $V: Q_1\times Q_2 \rightarrow \Re_+\cup \{\infty\}$ is a $\tau$-\textbf{simulation function of} $T_1$ \textbf{by} $T_2$ if 
for all $(q_1,q_2) \in Q_{1}\times Q_{2}$,
\begin{enumerate}
	\item[A0)] $V(q_1,q_2) \geq d_\Pi(\out{q_1},\out{q_2})$
	\item[A1)] $V(q_1,q_2) \geq \sup_{q_1 \trans{\slabel \in \labelSet} q_1'} \inf_{q_2 \trans{\slabel' \in B_\tau(\slabel)} q_2'} V(q_1',q_2')$	
\end{enumerate}
\end{definition}

A $\tau$-simulation function defines a $\teps$-STAS relation via its level sets. 
Namely, as shown in \cite[Thm. 3.4]{JuliusP_ApxSynchronizationMTS06}, the $\varepsilon$-sublevel set of $V$
\begin{equation}
\label{eq:levelset}
\Lc_\varepsilon^V = \{(q,q') \in Q_1 \times Q_2 \sut V(q,q') \leq \varepsilon\}
\end{equation}
is a $\teps$-STAS relation of $T_1$ by $T_2$ for all $\varepsilon \geq 0$. 

To keep the equations readable, in what follows, we define the following:
given $\slabel_{12} = (\slabel_1,\slabel_2)\in \labelSet_{12}$,
\[\Bttf \defeq \{(\slabel_3,\slabel_4) \in \labelSet_{34} \sut d_{\tlabelSet}({\tlabel}_{1},{\tlabel}_3) \leq \tau\}\]

($\labelSet_{34}$ is defined analogously to $\labelSet_{12}$ in \eqref{eq:SigmaEmbedding}).
The ball $\Bttf$ contains all labels in $\labelSet_{34}$ whose `chronological component' ${\tlabel}_{3}$ is no more than $\tau$-away from ${\tlabel}_1$.
Note that by definition for any $(\slabel_3,\slabel_4) \in \labelSet_{34}$, ${\tlabel}_3 = {\tlabel}_4$ (and analogously ${\tlabel}_1 = {\tlabel}_2$) so the above definition effectively bounds the distance between both chronological components of the label.

Consider the OMTS $T_1,T_2,T_3,T_4$, with $T_1$ in a feedback loop with $T_2$, and $T_3$ with $T_4$.
Let $\vot$ be a $\tau_{13}$-STAS function of $T_1$ by $T_3$ (Def.~\ref{def:stas fnt omts}),
and $\vtf$ be a $\tau_{24}$-STAS function of $T_2$ by $T_4$.
All systems share the same label set $\labelSet$.
We introduce the following functions to keep the equations manageable:
given $q_1' \in Q_1, q_3 \in Q_3, \slabel_i \in \labelSet$, define
\[\underbar{V}_{13}(q_1',q_3,\slabel_1) \defeq \inf_{q_3 \trans{\slabel_3 \in B_{\tau_{13}}(\slabel_1)} q_3'} \vot(q_1',q_3')\]
\[\bVtf(q_2',q_4,\slabel_2) \defeq \inf_{q_4 \trans{\slabel_4 \in B_{\tau_{24}}(\slabel_2)} q_4'} \vtf(q_2',q_4')\]

Consider $\bVot$: if we think of $T_3$ as trying to match $T_1$ transitions by minimizing $\vot$ over the label ball $B_{\tau_{13}}$, then $\bVot$ measures how well it does it.
Similarly for $\bVtf$.

Because STAS functions certify STAS relations via \eqref{eq:levelset}, the following theorem provides a way to build STAS functions for interconnections of systems, from the STAS functions of the individual connected systems.

\begin{theorem}
\label{thm:sgc}
Consider the OMTS $T_1,T_2,T_3,T_4$ with common label set $\labelSet$ interconnected as described above.
Let $\vot$ be a $\tau_{13}$-STAS function of $T_1$ by $T_3$,
and $\vtf$ be a $\tau_{24}$-STAS function of $T_2$ by $T_4$.
Set $\tau = \min(\tau_{13}, \tau_{24})$.

Define $V: Q_{12} \times Q_{34} \rightarrow \overline{\Re_+}$
to be $V((q_1,q_2),(q_3,q_4)) = h(\vot(q_1,q_3),\vtf(q_2,q_4))$ where $h$ is continuous and non-decreasing in both arguments.

Recall the definition of lifted label sets  $\labelSet_{12},\labelSet_{34}$ in \eqref{eq:lifting}.
Let $g:\Re \rightarrow \Re$ be a non-decreasing function 
s.t. $g(x) \geq x$ and 
for all $q_{12} \in Q_{12}$, $q_{34} \in Q_{34}$,
$g$ satisfies
\begin{eqnarray}
&&\sup_{\qot \trans{(\slabel_1,\slabel_2) \in \labelSet \times \labelSet} \qot'} h(\bVot(q_1',q_3,\slabel_1),\bVtf(q_2',q_4,\slabel_2)) \geq \; \nonumber \\
&& g\left(\sup_{\qot \trans{(\slabel_1,\slabel_2) \in \labelSet_{12}} \qot'} h(\bVot(q_1',q_3,\slabel_1),\bVtf(q_2',q_4,\slabel_2)) \right) 
\nonumber \\
\label{eq:g}
\end{eqnarray}

Also, let $\gamma_1,\gamma_2: \Re \rightarrow \Re_+$ be continuous non-increasing functions s.t. $\gamma_i(x) \leq x$, $i=1,2$,
 and 
for all $\slabel_{12} = (\slabel_1,\slabel_2) \in \labelSet_{12}$,
for all $(q_3,q_4) \in Q_{34}$, 
and all $(q_1',q_2') \in Q_{12}$
\begin{eqnarray}
\bVot(q_1',q_3,\slabel_1) \geq \gamma_1(\inf_{q_3 \trans{B_{\tau}^{34}(\slabel_{12})} q_3'}\vot(q_1',q_3')) 
\label{eq:gamma1}
\\
\bVtf(q_2',q_4,\slabel_2) \geq \gamma_2(\inf_{q_4 \trans{B_{\tau}^{34}(\slabel_{12})} q_4'}\vtf(q_2',q_4')) 
\label{eq:gamma}
\end{eqnarray}

If the following conditions hold:
\begin{enumerate}[(a)]
\item \label{ass:Vcontinuous}
	$V$ is continuous in the product topology of $Q_{12}\times Q_{34}$.
\item \label{ass:hhtilde}
For all $\qot \in Q_1 \times Q_2,\qtf \in Q_3 \times Q_4$,
\begin{equation}
V(\qot,\qtf) \geq d_\Pi(\out{\qot},\out{\qtf})
\label{eq:hhtilde}
\end{equation}
\item \label{ass:g dist h}
Function $g$ distributes over $h$, that is
\[g(h(x,x')) = h(g(x),g(x')) \; \forall x,x' \]
\item \label{ass:sgc} [Small Gain Condition]
For all $x \in \Re$,
\[g\circ \gamma_1(x) \geq x, \; g\circ \gamma_2(x) \geq x\]
\end{enumerate}
then 
$V$ is a $\tau$-STAS function of $T_1 \circ T_2$ by $T_3 \circ T_4$.
\exmend
\end{theorem}

Before proving the theorem, a few words are in order about its hypotheses.
A function $g$ satisfying \eqref{eq:g} always exists: by observing that $\labelSet_{12} \subset \labelSet \times \labelSet$, we see that $g$ can be taken to be the identity. 
A non-identity function quantifies how restrictive is the interconnection $T_1 \circ T_2$.
It does so by quantifying the difference between the full label set $\labelSet \times \labelSet$ available to the individual systems operating without interconnection (on the LHS of inequality \eqref{eq:g}), and the restricted label set $\labelSet_{12}$ available to them as part of the interconnection (on the RHS).

Similarly, functions $\gamma_1,\gamma_2$ satisfying \eqref{eq:gamma} always exist: we can take $\gamma_i$ to be identically zero.
These choices, however, are unlikely to be useful: we need $\gamma_i$ to quantify how restrictive is the interconnection $T_3 \circ T_4$.
They do so by quantifying the difference between the full label ball $B_{\tau_{13}}(\slabel_1) \times B_{\tau_{24}}(\slabel_2)$ available to the individual systems operating without interconnection, and the restricted label ball $B_{\tau}^{34}(\slabel_{12}) \subset B_{\tau_{13}}(\slabel_1) \times B_{\tau_{24}}(\slabel_2)$ available to them as part of the interconnection.
See Fig.\ref{fig:labelSets} for an illustration of the label sets.
\begin{figure}
\centering
\includegraphics[scale=0.5]{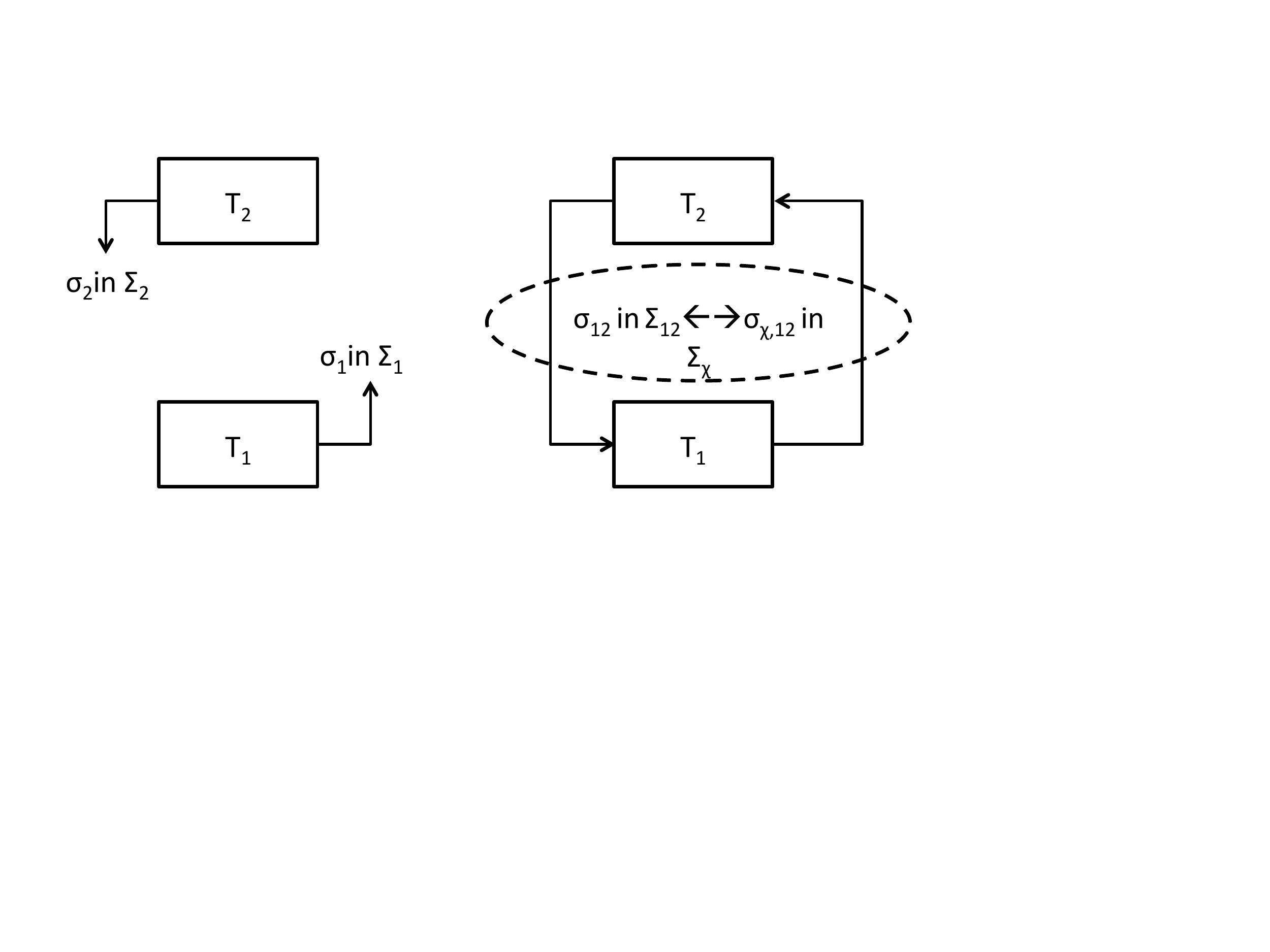}
\caption{Label sets constrained by interconnection. $\labelSet_{12}$ is the set of label pairs compatible with the interconnection as given in Def.\ref{def:feedbackOMTS}. }
\label{fig:labelSets}
\end{figure}

These two aspects are similar to the conditions, in more classical Lyapunov-based small gain theorems, placing a minimum on the rate of decrease of the Lyapunov functions of the individual systems, and that bound is related to the growth of the other system's Lyapunov function. (For example results on input-to-state stability \cite{JiangMW_LyapunovISS},\cite{SanFelice_IOSS14}, and for bisimulation functions in non-hybrid systems~\cite{Girard_CompositionBisim07}).
Now the more restrictive $T_1 \circ T_2$ is, the bigger $g$ can be.
The more restrictive $T_3 \circ T_4$ is, the smaller $\gamma_i$ need to be.
The Small Gain Condition (SGC) says that the restrictiveness of $T_1\circ T_2$ must be balanced by that of $T_3 \circ T_4$:
if $T_3 \circ T_4$ is too restrictive ($\gamma_i(x) << x$) relative to $T_1 \circ T_2$ ($g\circ \gamma_i(x) < x$), 
then $T_1 \circ T_2$ can play a label $\slabel_{12}$ that can't be matched, and thus we lose similarity of the systems.
Thus similar to the classical results (e.g.,~\cite{Girard_CompositionBisim07}), the SGC balances the gains of the feedback loops.


	\begin{proof} (Thm.~\ref{thm:sgc})
		
		We seek a STAS function $V: Q_{12} \times Q_{34} \rightarrow \overline{\Re_+}$ which would certify that $T_3 \circ T_4$ simulates $T_1 \circ T_2$, and we seek the corresponding precision $\teps$.
		
		For notational convenience, introduce
		\[\bV(\qot',\qtf,\slabel_{12}) \defeq \inf_{\qtf \trans{\slabel' \in B_{\tau}^{34}(\slabel_{12})} \qtf'} V(\qot', \qtf')\]

			By definition, $V$ must satisfy for all $(\qot,\qtf) \in Q_{12}\times Q_{34}$,
			\begin{enumerate}
				\item[A0)] $V(\qot,\qtf) \geq d_\Pi(\out{\qot},\out{\qtf})$
				\label{item:V0}
				\item[A1)] \begin{eqnarray*}
				V(\qot,\qtf) \geq \sup_{\qot \trans{\slabel \in \labelSet_{12}} \qot'} (\inf_{\qtf \trans{\slabel' \in B_\tau^{34}(\slabel)} \qtf'} V(\qot',\qtf'))
				\end{eqnarray*}	
				\label{item:V1}
			\end{enumerate}
		
		Condition A0 is the same as \eqref{eq:hhtilde}, and so is true by hypothesis.	
		Now for A1. 
		First we restate it using $\bV$:
		\begin{eqnarray*}
			V(\qot,\qtf) \geq \sup_{\qot \trans{(\slabel_1,\slabel_2) \in \labelSet_{12}} \qot'} \bV(\qot',\qtf,\slabel)
		\end{eqnarray*}
		For all $q_1,q_2,q_3,q_4$,
		\begin{eqnarray*}
			&&h(\vot(q_1,q_3), \vtf(q_2,q_4)) \\
			&\geq& h(\sup_{\labelSet_1} \bVot(q_1',q_3,\slabel_1), \sup_{\labelSet_2} \bVtf(q_2',q_4,\slabel_2)  \\
			&\geq&\sup_{\labelSet_1} \sup_{\labelSet_2} h(\bVot(q_1',q_3,\slabel_1),   \bVtf(q_2',q_4,\slabel_2))  \\
			&=&\sup_{(\slabel_1,\slabel_2) \in \labelSet_1 \times \labelSet_2} h(\bVot(q_1',q_3,\slabel_1),   \bVtf(q_2',q_4,\slabel_2))  
		\end{eqnarray*}
		where we used property A1 for $\vot$ and $\vtf$ and the fact that $h$ is non-decreasing to obtain the first inequality, and the non-decreasing nature of $h$ to obtain the second inequality.
		(The second inequality becomes equality if $\bVot$ and $\bVtf$ achieve their suprema over $\labelSet_1$ and $\labelSet_2$ respectively.)
		Using \eqref{eq:g}, it comes
		\begin{eqnarray*}
		&&h(\vot,\vtf) \geq \\
		&& g(\sup_{\qot \trans{(\slabel_1,\slabel_2) \in \labelSet_{12}} \qot'} h(\bVot(q_1',q_3,\slabel_1),   \bVtf(q_2',q_4,\slabel_2)) )\\
		\end{eqnarray*}
		
		Applying \eqref{eq:gamma1},\eqref{eq:gamma} to the RHS of this last inequality, 
		\begin{eqnarray*}
			&&h(\vot,\vtf) \geq \\
			&& g(\sup_{\qot \trans{(\slabel_1,\slabel_2) \in \labelSet_{12}} \qot'} h(\gamma_1(\inf_{\Bttf}\vot),\gamma_2(\inf_{\Bttf}\vtf) )\\
		\end{eqnarray*}
		where we are using $\inf_{\Bttf}V_{ij}$ as an abbreviation for 
		\[\inf_{q_j \trans{\Bttf}q_{j}'}V_{ij}(q_i',q_j')\]
		We now establish two inequalities.
		First, note that 
		\begin{equation}
		\label{eq:gamma1decreasing}
		\gamma_1(\inf_{\Bttf}\vot) \geq \inf_{\Bttf} \gamma_1(\vot)
		\end{equation}
		Indeed, let 
		\begin{equation*}
		\label{eq:Q3bar}
		\bar{Q}_3= \text{Post}(q_3,\Bttf)
		\end{equation*}
		be the set over which the infimization is happening.
		We have that $v_* \defeq \inf_{\Bttf} \vot(q_1',q_3')$ is finite since $V$ is lower bounded by 0. 
		Now since $v_* \leq v$ for all $v \in \vot(\bar{Q}_3)$,
		and $\gamma_1$ is non-increasing, it follows that $\gamma_1(v_*) \geq \gamma_1(v)$ for all $v \in \vot(\bar{Q}_3)$. Taking the infimum on the RHS, the inequality \eqref{eq:gamma1decreasing} follows. 
		An inequality analogous to \eqref{eq:gamma1decreasing} holds for $\gamma_2$ by a similar argument.
		
		Second, note that because $\gamma_i$ and $V$ are continuous, and $\bar{Q}_3$ is compact, then the set $\gamma_i\circ V(\bar{Q}_3)$ is compact as well. 
		Since $h$ is continuous as well, it achieves its infimum over compact sets and therefore 
		\begin{eqnarray}
		\label{eq:h achieves inf}
		&&h(\inf_{B_\tau^{34}(\slabel_{12})} \gamma_1(\vot), \inf_{B_\tau^{34}(\slabel_{12})} \gamma_2( \vtf))\nonumber  \\
		&&= \inf_{B_{\tau}^{34}(\slabel_{12})} h(\gamma_1 \circ \vot,  \gamma_2 \circ \vtf)
		\end{eqnarray}
		
		We can proceed as
		\begin{eqnarray*}
			h&&(\vot,\vtf) \\
			&\geq & g(\sup_{\qot \trans{\slabel_{12}} \qot'}h(\gamma_1(\inf_{B_\tau^{34}(\slabel_{12})} \vot), \gamma_2(\inf_{B_\tau^{34}(\slabel_{12})} \vtf))) \\
			&\geq &g(\sup_{\qot \trans{\slabel_{12}} \qot'}h(\inf_{B_\tau^{34}(\slabel_{12})} \gamma_1(\vot), \inf_{B_\tau^{34}(\slabel_{12})} \gamma_2( \vtf))) \\
			& = & g(\sup_{\qot \trans{\slabel_{12}} \qot'} \inf_{B_{\tau}^{34}(\slabel_{12})} h(\gamma_1 \circ \vot,  \gamma_2 \circ \vtf)) \\
			& = & \sup_{\qot \trans{\slabel_{12}} \qot'} \inf_{B_{\tau}^{34}(\slabel_{12})} g\circ h(\gamma_1 \circ \vot,  \gamma_2 \circ \vtf)
		\end{eqnarray*}
		To obtain the second inequality, we used \eqref{eq:gamma1decreasing} and the fact that $h$ and $g$ are non-decreasing.
		To obtain the equalities, we used \eqref{eq:h achieves inf} and the fact that $g$ is non-decreasing.
		
		By distributivity of $g$ over $h$ and the SGC
		\begin{eqnarray*}
			&&h(\vot,\vtf) \nonumber \\
			&&\geq \sup_{\qot \trans{\slabel_{12}} \qot'} \inf_{B_{\tau}^{34}(\slabel_{12})}  h(g\circ \gamma_1 \circ \vot,  g\circ \gamma_2 \circ \vtf) \nonumber \\
			&& \geq \sup_{\qot \trans{\slabel_{12}} \qot'} \inf_{\qtf\trans{B_{\tau}^{34}(\slabel_{12})}\qtf'}  h( \vot,   \vtf) \\
			\label{eq:Vissimu}
		\end{eqnarray*}
		thus concluding that $V = h(\vot,\vtf)$ satifies A1, and so is a $\tau$-STAS function.
	\end{proof}

\textbf{About the other conditions} The distributivity assumption in (\ref{ass:g dist h}) holds, for example, if $h$ is the max operator, i.e. $h(x,x')=\max(x , x')$.

Thm.~\ref{thm:sgc} assures us that feedback interconnection respects similarity relation, and therefore also respects conformance relations.

However, the conditions defining $g$ and $\gamma_i$ (equations \eqref{eq:g} and \eqref{eq:gamma},\eqref{eq:gamma1}) are technical conditions that are are hard to check. 
Turning them into a computational tool for particular classes of systems is the subject of current research.
A simpler, and more conservative, criterion is given in the following theorem:
\begin{theorem}
\label{thm:suff conf for gamma}
If 
\[k_1 \defeq \frac{\inf_{Q_1} \inf_{Q_3}\vot(q_1,q_3)}{\sup_{Q_3} \sup_{Q_1}\vot(q_1,q_3)}  < \infty\]
then $\gamma_1(v) = k_1v$ satisfies \eqref{eq:gamma1}.
Similarly, if 
\[k_2 \defeq \frac{\inf_{Q_2} \inf_{Q_4}\vtf(q_2,q_4)}{\sup_{Q_2} \sup_{Q_4}\vtf(q_2,q_4)}  < \infty\]
then $\gamma_2(v) = k_2v$ satisfies \eqref{eq:gamma}.
\exmend
\end{theorem}
\begin{proof}
	We give the proof for $k_1$, that for $k_2$ is similar.
	Define $\bar{Q}_3= \{q_3' \sut \exists q_3 \trans{B_{\tau}^{34}(\slabel_{12})} q_3'\}$ 
	and $\hat{Q}_3= \{q_3' \sut \exists q_3 \trans{B_{\tau_{13}}(\slabel_1)} q_3'\}$.
	Since $\proj{B_{\tau}^{34}(\slabel_{12})}{\labelSet} \subset B_{\tau_{13}}(\slabel_1)$, $\bar{Q}_3 \subset \hat{Q}_3 \subset Q_3$.
	Thus for any $q_1' \in Q_1$
	\begin{eqnarray*}
		\inf_{q_3' \in \hat{Q}_3}&&\vot(q_1',q_3') \geq \inf_{q_3' \in Q_3}\vot(q_1',q_3') 
		\\
		&\geq& \inf_{q_3' \in Q_3}\vot(q_1',q_3') \frac{\inf_{\bar{Q}_3}\vot(q_1',q_3')}{\sup_{Q_3}\vot(q_1',q_3')}
		\\
		&\geq& \frac{\inf_{(q_1,q_3') \in Q_1 \times Q_3}\vot(q_1,q_3')}{\sup_{(q_1,q_3') \in Q_1 \times Q_3}\vot(q_1,q_3')} \inf_{\bar{Q}_3}\vot(q_1',q_3')
		\\
		&= & k_1 \inf_{q_3' \in \bar{Q}_3}\vot(q_1',q_3')
	\end{eqnarray*}
\end{proof}
The challenge with the choice of $\gamma_1$ and $\gamma_2$ in Thm. \ref{thm:suff conf for gamma} is that $g$ is now required to always `compensate' for the worst-case behavior to satisfy the SGC. 
I.e. we need $g(x) \geq x/\max(k_1, k_2)$ for all $x$.
This may lead to a violation of \eqref{eq:g}. 

The next result follows from Thm.\ref{thm:sgc}, the fact that $h$ is increasing, and \cite[Thm. 3.6]{JuliusP_ApxSynchronizationMTS06}.
\begin{theorem}
Let 
$\varepsilon_{13} = \sup_{Q_1^0} \inf_{Q_3^0} \vot(q_1,q_3)$
and $\varepsilon_{24} = \sup_{Q_2^0} \inf_{Q_4^0} \vtf(q_2,q_4)$,
so that $T_3$ $(\tau_{13},\varepsilon_{13})$-simulates $T_1$,
and $T_4$ $(\tau_{24},\varepsilon_{24})$-simulates $T_2$.
Then $T_3\circ T_4$ $(\tau,\varepsilon)$-simulates $T_1 \circ T_2$ with 
$\tau = \min(\tau_{13}, \tau_{24}), \varepsilon = h(\varepsilon_{13},\varepsilon_{24})$.
\end{theorem}

\section{Related works}
\label{sec:relatedWork}

In this paper we understand conformance as a notion that relates \emph{systems}, as done in~\cite{TalpinGSG_conformance05}, rather than a system and its specification as done for example in~\cite{DangN09_CovGuidedTestGen}.
Most existing works on system conformance, either requires equality of outputs, or does not account for timing differences, as in~\cite{MajumdarSUY_CompositionalEC13} where an approximate method for verifying formal equivalence between a model and its auto-generated code is presented.
The approach to conformance of Hybrid Input/Output Automata in~\cite{Osch_IOCO06} and falls in the domain of nondeterministic abstractions, and a thorough comparison between this notion and ours is given in \cite{MohaqeqiMT14_ComparativeStudy}.
The works closest to ours are~\cite{JuliusP_ApxSynchronizationMTS06} and~\cite{Quesel_SimilarityGamesThesis13}.
The work \cite{JuliusP_ApxSynchronizationMTS06} defines the STAS relation we used in this paper. 
The goal in~\cite{JuliusP_ApxSynchronizationMTS06} is to define robust approximate synchronization between systems (rather than conformance testing). 
The refinement relation between systems given in~\cite{QueselFD11_CrossingBridge} allows different inputs to the two systems.
Conformance requires the same input be applied, which is a more stringent requirement.
The current theoretical framework also  allows a significantly wider class of systems than in~\cite{QueselFD11_CrossingBridge}. 
\section{Conclusions}
When a system model goes through multiple design and verification iterations, it is necessary to get a rigorous and quantitative measure of the similarities between the systems. 
Conformance testing \cite{AbbasMF_MemocodeConformance14} allows us to obtain such a measure, and to automatically transfer formal verification results from a simpler model to a more complex model of the system.
In this paper, we extended the reach of conformance testing by developing the sufficient conditions for feedback interconnections of conformant systems to be conformant. 
As pointed out earlier, these conditions apply to Open Metric Transition Systems, and while this means they are very broadly applicable, they must be specialized to specific classes of dynamical systems.
The next step is to compute STAS functions for various classes of dynamial systems, including hybrid systems. 
This is the subject of current research.
In addition, we aim to apply the compositionality theory developed here to problems in source code generation.



\bibliographystyle{abbrv}
\bibliography{CDC2015,conformance,conformanceMEMOCODE,HSCC2015_CompositionalConf,fainekos_bibrefs,SAREDnonlinear_2014}

\end{document}